\documentclass[aps,pra,showpacs,twoside,10pt,floatfix,twocolumn]{revtex4-2}
\usepackage[colorlinks=true, citecolor=blue, urlcolor=blue ]{hyperref}
\usepackage{epsfig,newlfont,amssymb,amsfonts,amsmath,bm,subfigure,palatino,mathtools,amsthm,braket,soul,enumitem,xcolor,graphics,graphicx,times,physics}
\usepackage[normalem]{ulem}
\newtheorem{theorem}{Theorem}
\newtheorem{corollary}{Corollary}

\begin{document}
\title{
 Critical quantum metrology using non-Hermitian spin model with $\mathcal{RT}$-symmetry \\
}
\author{Keshav Das Agarwal$^1$,  Tanoy Kanti Konar$^{1}$, Leela Ganesh Chandra Lakkaraju$^{1,2,3}$, Aditi Sen(De)$^{1}$}
\affiliation{$^1$ Harish-Chandra Research Institute, A CI of Homi Bhabha National Institute, Chhatnag Road, Jhunsi, Allahabad - 211019, India}
\affiliation{$^2$ Pitaevskii BEC Center and Department of Physics, University of Trento, Via Sommarive 14, I-38123 Trento, Italy }
\affiliation{$^3$ INFN-TIFPA, Trento Institute for Fundamental Physics and Applications, Trento, Italy} 

\begin{abstract}
The non-Hermitian transverse \(XY\) model with Kaplan-Shekhtman-Entin-Wohlman-Aharony (KSEA) interaction having \(\mathcal{RT}\)-symmetry, referred to as \(iKSEA\) model, possesses both an exceptional point at which eigenvectors coalesce and a quantum critical point where gap-closing occurs. To precisely estimate the magnetic field of the system, we prove that the quantum Fisher information (QFI) of the ground state of the \(iKSEA\) model, which is a lower bound of the precision quantified by the root mean square error, scales as \(N^2\), with \(N\) being the system-size. This provides Heisenberg scaling both at the quantum critical point and the exceptional point in the thermodynamic limit. It indicates that reservoir engineering can provide enhanced precision of system parameters when the system is in contact with the bath, resulting in this non-Hermitian model. 
Additionally, we demonstrate analytically that, in contrast to Hermitian systems, QFI surpasses the Heisenberg limit and achieves super-Heisenberg scaling (\(\sim N^6\)), when the strength of the KSEA interaction approaches the anisotropy parameter, permitting competition between non-hermiticity and hermiticity features, as long as the system size is moderate. Moreover, we illustrate that starting from a product state, the non-Hermitian evolving Hamiltonian can create the dynamical state that surpasses the standard quantum limit in the broken regime.

\end{abstract}

\maketitle 

\section{Introduction}
\label{sec:intro}


 A significant milestone in the development of quantum technology has been the discovery of quantum sensors, which offer advantages over existing classical devices in the precision measurement of parameters such as temperature, magnetic, electric, and gravitational fields, and time  \cite{sensing_review_1,sensing_review_2,montenegro2024review}.  Further, it was found that  the sensitivity of  system parameters can be enhanced by utilizing cooperative quantum phenomena, like quantum critical point,  acting as a resource in quantum sensing  \cite{zanardi_1,zanardi_2,zanardi_3}. 
 In order to establish quantum benefit, a crucial metric used is quantum Fisher information (QFI), where  a greater QFI indicates a more accurate system parameter estimation \cite{Cramr1946}. 
 Precisely, when employing an \(N\)-particle interacting system as a probe, QFI scales as \(\sim N^\mu\) with \(\mu\) being the scaling component \cite{giovannetti_prl_2006,Giovannetti2011Apr}. The standard quantum limit (SQL) \cite{pezze_prl_2009}, which can be attained by conventional means, is represented by the situation of \(\mu=1\) 
 while when \(\mu>1\), the system exhibits a quantum advantage. Notably, when \(\mu\sim 2\), the scaling reaches the Heisenberg limit (HL). Interestingly, certain quantum systems can surpass the HL, achieving what is referred to as the super-Heisenberg limit, and therefore displaying an even higher quantum advantage \cite{boixo_prl_2007,boixo_prl_2008,mishra_prl_2021,Yousefjani2023Oct,Mihailescu2024May,sahoo_pra_2024,Adani2024Aug,Mondal2024Jul,Mihailescu2024Jul,Yousefjani_pra_2025,Mihailescu2025Mar}.

In contrast, significant advancements have been made in non-Hermitian systems, which provide a framework for controlling a system when it interacts with its environment \cite{bender_prl_1998,bender_prl_2002,bender_ropp_2007, Ueda_review,metelmann_prx_2015, Fazio_reservoir_engineering_2018}.
 Since the inception of this notion, extensive research has been conducted to explore the topological and physical properties of such non-equilibrium systems \cite{lee_prx_2014,gong_prx_2018,kawabata_prx_2019,Martinez_prb_2018,Borgnia_prl_2020,Li2020Oct,Chen2024Nov}. Additionally, recent experiments have confirmed that non-Hermitian Hamiltonians can be  realized in certain situations  \cite{guo_pt_exp,Wu2019May,Naghiloo2019Dec} 
such as when quantum jumps in a Markovian quantum master equation \cite{open_quan_book} are omitted or when continuous measurements on a quantum system have a no-click limit \cite{Roccati2022Mar}. These developments have had a significant influence on the  study of non-Hermitian systems in various quantum technologies, including quantum batteries \cite{konar2022quantum}, quantum sensing \cite{Chen2019Aug,budich_prl_2020,Edvardsson_prb_2022,Sarkar2024Jul, Sarkar2024Nov}, and thermal machines \cite{Santos2023Dec}. 

In the context of quantum sensing, exceptional points (EPs) \cite{bender_ropp_2007}, where the eigenspectrum coalesces, are particularly significant since, in their vicinity,  the scaling of quantum Fisher information gets improved \cite{liu_prl_2016,chen_nature_ep_sensing_2017,Lau2018Oct,McDonald2020Oct,chu_prl_2020}.
Nonetheless, there are ongoing discussions about whether QFI enhancement around EPs is legitimate in critical quantum metrology  \cite{ding_prl_2023}, even though it has been shown that quantum sensing can still function well without EPs \cite{budich_prl_2020,bao_prap_2022,Sarkar2024Jul}. 
The fundamental reason why a non-Hermitian system acts as a quantum sensor lies in the shift of the energy spectrum at the EP. When the parameters of the Hamiltonian vary infinitesimally by $\epsilon$, the resulting change in energy $\delta \omega$ follows as $\sqrt{\epsilon}$
while the response to perturbation at the critical point in a Hermitian quantum sensor follows a linear scaling
($\delta \omega \sim \epsilon$) \cite{pt_sym_exp1}. A particularly striking feature emerges when considering non-Hermitian many-body systems, QFI scales exponentially with the order of the exceptional point~\cite{pt_sym_exp1}. This results in a substantially stronger response to perturbations of system parameters compared to Hermitian many-body systems. Such enhanced sensitivity has been experimentally demonstrated in various physical platforms, including superconducting qubits \cite{superconducting_nonH_sensor}, trapped ions \cite{claverorubio2025vibrationalparametricarraystrapped, 10.1063/5.0168372}, optical cavity \cite{Chen2017Aug}, photonic systems \cite{Guo2021Apr,xiao_prl_2024,Yu2024May,Parto2025Jan} and NV centers \cite{wu2025enhancedquantumsensingtimemodulated_nitrogen}, among others.  
Therefore, it is essential to comprehend the significance of critical points in non-Hermitian Hamiltonians in order to develop quantum sensors, much like in the Hermitian situation.

We employ here a non-Hermitian  \(XY\) model along with  Kaplan-Shekhtman-Entin-Wohlman-Aharony (KSEA) interactions  to answer the question of accurately estimating the magnetic field. It possesses rotational and time symmetry (\(\mathcal{RT}\)-symmetry) and is  referred  to as \(iKSEA\) model. This choice is motivated by the fact that it undergoes two quantum phase transitions -- the critical point where the gapped to gapless transition occurs and the exceptional point in which eigenvalues coalesce and the system moves from the unbroken phase (the eigenspectrum remains real, and the eigenvectors commute with the \(\mathcal{RT}\)-symmetric operator) to the broken one (the eigenvalues become complex, and the eigenvectors no longer commute with the \(\mathcal{RT}\)-symmetric operator). 
Importantly, it can be effectively simulated  by designing a suitable bath and ignoring the jump terms in the Markovian master equation \cite{open_quan_book, Agarwal2023May}.  We demonstrate the significance of developing quantum sensors based on non-Hermitian systems by analytically proving that both exceptional and quantum critical points in the ground state of the \(iKSEA\) model  can greatly increase QFI, yielding the Heisenberg limit (HL).    Furthermore, we demonstrate that when the system size is moderate — achievable with current technology —and the KSEA interaction strength approaches the anisotropy parameter, QFI outperforms HL and scales as \(N^6\) with \(N\) being the system-size. It is important to note here that this super-Heisenberg scaling can only be obtained due to the contention between non-Hermitian and Hermitian parameters in the \(iKSEA\) model.



In the dynamical framework, both the time and system-size are considered as resource for quantum metrology when the system evolved from a suitable initial state. By choosing the initial state as a product state, sensing a parameter, encoded by a $k$-body operator provides QFI which scales as $N^k t^2$ \cite{boixo_prl_2007, puig2024arxiv} for the optimal  Hermitian evolution while better scaling of QFI has been shown starting from entangled states \cite{pang_pra_2014, sahoo_pra_2024, abiuso_prl_2025}, or with stark localization \cite{he_prl_2023, manshouri2024} or long-range Hamiltonian \cite{guan_prr_dqpt_sensing_2021, shi_prl_2024}.  Our study exhibits that the quantum Fisher information  of the dynamical state, used to estimate the magnetic field, can exceed the standard quantum limit  at long evolution times, thereby demonstrating a genuine quantum advantage. This enhancement arises when the system is initialized in a product state and subsequently evolves under a non-Hermitian Hamiltonian in the broken phase. Importantly, this advantage cannot be replicated by any equivalent nearest-neighbor Hermitian model with uniform magnetic field.

The organization of the paper is as follows: After the Introduction, we briefly introduce the parameter to be estimated and the corresponding non-Hermitian Hamiltonian in Sec. \ref{sec:setstage}. In Sec. \ref{sec:criticalpointscaling}, we prove analytically that Heisenberg scaling of \(N^2\) with system-size, \(N\) can be achieved both at the quantum critical points and at the exceptional points, thereby showing the benefit of non-Hermitian model.  In Sec. \ref{sec:dynamicsexceptionalpoint}, we show, beyond the static scenario, how nonlinear scaling of quantum Fisher information with time can be attained starting from the product state and after tuning the system parameters. Sec. \ref{sec:conclu} summarizes the results.


\section{Set the stage: Quantum Fisher Information and non-Hermitian Model }
\label{sec:setstage}

For non-Hermitian system, we first discuss about  the Cram{\'e}r - Rao bound which is a lower bound for the estimation of  the parameters.  We then introduce the non-Hermitian model that we consider in this work and briefly describe its critical points. 

\subsection{ Quantum Fisher information for non-Hermitian model.}
In order to estimate an unknown parameter \( \theta \),  the parameter is encoded on a state \( \rho_\theta \), referred to as the probe system on which 
positive operator-valued measurements (POVM) is performed to infer \(\theta\). The 
Cramér-Rao bound \cite{Cramr1946},   \( \delta\theta \geq \frac{1}{\sqrt{M \mathcal{F}^H_\theta}}, \) provides  the lower bound on the precision of the estimation of \( \theta \),
where \(\mathcal{F}^H_\theta\) represents the QFI, obtained by maximizing over all possible measurement bases~\cite{sensing_review_1,sensing_review_2}. 
While optimizing over measurements is necessary to determine the QFI for a general state, for pure states, it can be computed directly as 
 \(   \mathcal{F}^{H}_\theta(\ket{\psi}) = 4  \langle d_\theta\psi | d_\theta\psi \rangle - |\langle d_\theta\psi | \psi \rangle|^2 \),
where \( | d_\theta\psi\rangle = \frac{d}{d\theta}|\psi\rangle \), and the superscript, \(H\) denotes the Hermitian system. 

In non-Hermitian systems, the eigenstates do not form an orthonormal basis, as the right and left eigenvectors of the Hamiltonian are not identical. Hence, to construct a proper probability distribution, it is necessary to renormalize the eigenvectors using the Dirac norm formulation 
as \(\ket{\Phi_n} \rightarrow \frac{\ket{\Phi_n}}{\sqrt{\braket{\Phi_n }{ \Phi_n}}}\), where \( \ket{\Phi_n} \) represents a right eigenstate of the Hamiltonian and the corresponding  
QFI reads as 
\begin{equation}
    \mathcal{F}^{nH}_\theta(\ket{\Phi_n}) = 4 \text{Re} \left( \langle d_\theta\Phi_n | d_\theta\Phi_n \rangle - |\langle d_\theta\Phi_n | \Phi_n\rangle|^2 \right),
    \label{eq:qfi_def}
\end{equation}   
where \(nH\) in the superscript represents the non-Hermitian system and  its derivative is also evaluated with respect to the Dirac inner product.

\subsection{ Non-Hermitian quantum spin model with \(\mathcal{RT}\)-symmetry} 

Let us consider the non-Hermitian \(XY\) model in conjunction with Kaplan-Shekhtman-Entin-Wohlman-Aharony ($iKSEA$) model on $N$ spin-$1/2$ systems, given by 
\begin{align}
    \nonumber H^{iKSEA} &= \sum_{j=1}^N \frac{(1+i\gamma)}{4} \sigma_j^x\sigma_{j+1}^x + \frac{(1-i\gamma)}{4} \sigma_j^y\sigma_{j+1}^y  \\ 
    & + \frac{K}{4}(\sigma_j^x\sigma_{j+1}^y+\sigma_j^y\sigma_{j+1}^x) + \frac{h}{2} \sigma_j^z,
    \label{eq:Hamil}
\end{align}  
with the periodic boundary condition, i.e. $\sigma^a_{L+1}\equiv \sigma^a_1$, where \( \sigma^a \) (\( a \in \{x, y, z\} \)) are the Pauli matrices, \( \gamma \) represents the anisotropy parameter responsible for the non-Hermiticity, and \( K \) and \( h \) denote the strength of the $KSEA$ interaction and the external magnetic field, respectively. This model possesses rotational and time-reversal symmetry, together called \(\mathcal{RT}\)-symmetry which is given as \(\mathcal{R}\equiv \exp[-i\frac{\pi}{4}\sum_{j=1}^N\sigma_j^z]\), a \(\frac{\pi}{2}\) rotation around the \(z\)-axis and the time-reversal operation, \(\mathcal{T}i\mathcal{T}^{-1}=-i\). Although the Hamiltonian satisfies the relation, \([H,\mathcal{RT}]=0\), due to the anti-linear properties of the symmetry operations in the broken region, eigenstates of the Hamiltonian do not have \(\mathcal{RT}\) symmetry.

A key motivation for choosing the \( H^{iKSEA} \) model is that it can be solved analytically using a combination of the Jordan-Wigner, Fourier, and Bogoliubov transformations~\cite{barouch_pra_1970,barouch_pra_1971} and possess rich phases \cite{Agarwal2023May} (see Appendix \ref{appendix:exactdiag}). The system exhibits both unbroken and broken phases within the \((h, \gamma, K)\) parameter space, separated by exceptional points. In the unbroken phase, the eigenvalues are real, and the eigenvectors form a complete (though non-orthogonal) basis that satisfies $\mathcal{RT}$ symmetry, and can also be biorthogonalized. In the broken phase, the Hamiltonian becomes defective with an incomplete set of eigenvectors, leading to the breaking of $\mathcal{RT}$ symmetry. At the exceptional points, the eigenvalues remain real (and degenerate), but the non-orthogonal eigenvectors coalesce, rendering the Hamiltonian defective.
The \( iKSEA \) model features a critical magnetic field strength at \( h_c = 1 \), where the ground state becomes degenerate without the eigenvectors coalescing, similar to traditional Hermitian systems. 
When \( K < \gamma \), the exceptional point occurs at \( h_e = \sqrt{1+\gamma^2-K^2} \), with the system being in the unbroken phase for \( h > h_e \) and in the broken phase for \( h < h_e \), while \( h_c = 1 \) remains a critical point in the broken phase. Interestingly, when \( K > \gamma \), the system remains entirely in the unbroken phase for all \( h \), with \( h_c =  1 \) still existing in the unbroken regime. Also, when $\gamma=K$, $h_c=1$ is a critical point, while $h_e=h<1$ are the exceptional points.
Importantly, this non-Hermitian Hamiltonian can be physically realized through reservoir engineering in the no-jump limit of the Markovian master equation \cite{open_quan_book}, where the system is described by an $XX$ Hamiltonian with $KSEA$ interaction \cite{kaplan_jpb_1983,shekhtman_prl_1992}.

 \begin{figure}
    \centering
    \includegraphics[width=\linewidth]{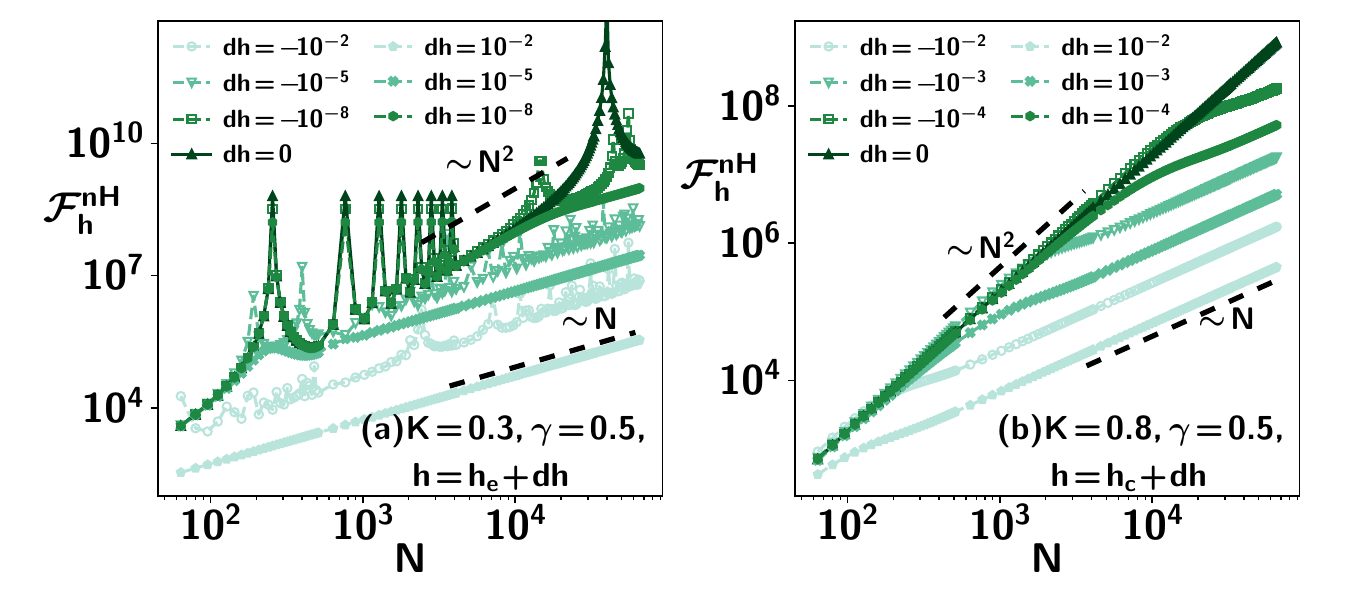}
    \caption{\textbf{Scaling of Fisher information around critical points.}  QFI (ordinate) against the system-size (abscissa) for different magnetic field strength around the exceptional point \(h_e\) (a) and the critical point \(h_c\) (b). Around both critical points, the system provides quadratic scaling with system-size provided the system-size is very high, approaching thermodynamic limit as shown in Theorem 1. All the axes are dimensionless.} 
    \label{fig:qfi_c}
 \end{figure}

\section{Critical and exceptional point quantum sensing: super-Heisenberg scaling }
\label{sec:criticalpointscaling}

It was demonstrated that when the system 
undergoes transitions from a gapped to a gapless phase by tuning a system parameter like the magnetic field,  \(h\), resulting in long-range correlations,  QFI, \(\mathcal{F}_{\theta=h}^H\), scales as \(N^2\), with \(N\) being the system-size. This naturally raises the question: ``Can non-Hermitian systems provide the same scaling as the Hermitian one?". Here we answer this question affirmatively by demonstrating that the Heisenberg scaling can be retained
not only at the quantum critical point but also at the exceptional point in a non-Hermitian system.

Before establishing the above result, 
let us first note that  defining the ground state of such a non-Hermitian model possesses challenges due to the presence of imaginary eigenvalues,  although   it is straightforward to define in the unbroken regime. 
In the broken phase, since eigenvalues may acquire imaginary parts, 
we select the eigenstate associated with the eigenvalue that has the maximum imaginary part. This choice is justified because such a state represents the steady state of non-Hermitian evolution, given that, in the long-time limit ($t \to \infty$), the system collapses into this state.


\begin{theorem} The quantum Fisher information in the ground state of the non-Hermitian \(iKSEA\) model exhibits Heisenberg scaling (grows quadratically with the increase of system-size)  at  the critical  and exceptional points in the thermodynamic limit.
\end{theorem}

\begin{proof}
After the application of the Jordan-Wigner and Fourier transformations, the Hamiltonian decouples into the momentum blocks, i.e., $H^{iKSEA}=\bigoplus_{p=0}^{N/2-1}\mathcal{H}_p$, with 
\begin{equation}
    \mathcal{H}_p=\left[\begin{array}{cc}
-g_p & -\alpha^+_p\\
\alpha^-_p & g_p
\end{array}\right];\quad \begin{array}{c}
g_p = h+\cos\phi_p,\\
\alpha^{\pm}_p= (\gamma\pm K)\sin\phi_p,
\end{array}
\label{eq:ksea_P}
\end{equation}
in the $\{\ket{0}, c_p^\dagger c_{-p}^\dagger\ket{0}\}$ basis. Here $c_p^\dagger$ and $c_p$ are fermionic creation and annihilation operators in the momentum basis, respectively, with $\phi_p=(2p-1)\pi /N $ and \(p \in [1, N/2]\) in the even parity sector (where the ground state of the finite system size is), and the corresponding eigenvalues are $\pm\epsilon(\phi_p)$ where \(\epsilon(\phi_p)=\sqrt{g_p^2-\alpha_p^+\alpha^-_p}\). In order to showcase the scaling behavior of Fisher information of the ground state, we perform a qualitative calculation by expanding the Fisher information around the critical points and show that the highest contribution changes as a square of the system size. To do this, we first perform
the Taylor expansion of the dispersion relation, especially \(\epsilon^2(x)\) around a particular point, \(x_0\) which is given as
\begin{eqnarray}
   && \epsilon^2(x) 
    = \epsilon(x_0)^2-2\sin x_0 (h+\tilde{h}\cos x_0)(x-x_0) \nonumber\\
    &&- (\tilde{h}\cos2x_0+h\cos x_0)(x-x_0)^2 + O((x-x_0)^3)
    \label{eq:expansion}
\end{eqnarray}
with $\tilde{h}=1+\gamma^2-K^2$.
The state corresponding to the energy $-\epsilon(\phi_p)$ is  $|\psi^-\rangle_p=\frac{1}{\sqrt{\mathcal{A}_p^-}}[u_p^-, v_p^-]^T$, with  $u^-_p=\alpha^+_p, v_p^-=\epsilon(\phi_p)-g_p$ and $\mathcal{A}_p^-$ is the Dirac normalization constant. The  ground state for \(H^{iKSEA}\) is $\Psi(h,\gamma,K)=\bigoplus_p|\psi^-\rangle_p$, and the corresponding Fisher information is obtained  by  summing over the Fisher information from each block, i.e., $\mathcal{F}_h^{nH}(\ket{\Psi})=\sum_p\mathcal{F}_h^{nH}(\ket{\psi^-}_p)$. Note that, if either each $\mathcal{F}_h^{nH}(\ket{\psi^-}_p)$ is bounded with system size $N$ (i.e. finite as $N\to \infty$), or if $\mathcal{F}^{nH}_h(\ket{\psi^-}_p)\sim N$ for a finite number of momentum indexes $p$, in that case,  QFI \(\sim N\) while $\mathcal{F}_h^{nH}(\ket{\psi^-}_p)\sim N^2$ implies Heisenberg scaling of QFI for the ground state. 

\begin{figure}
    \centering
    \includegraphics[width=\linewidth]{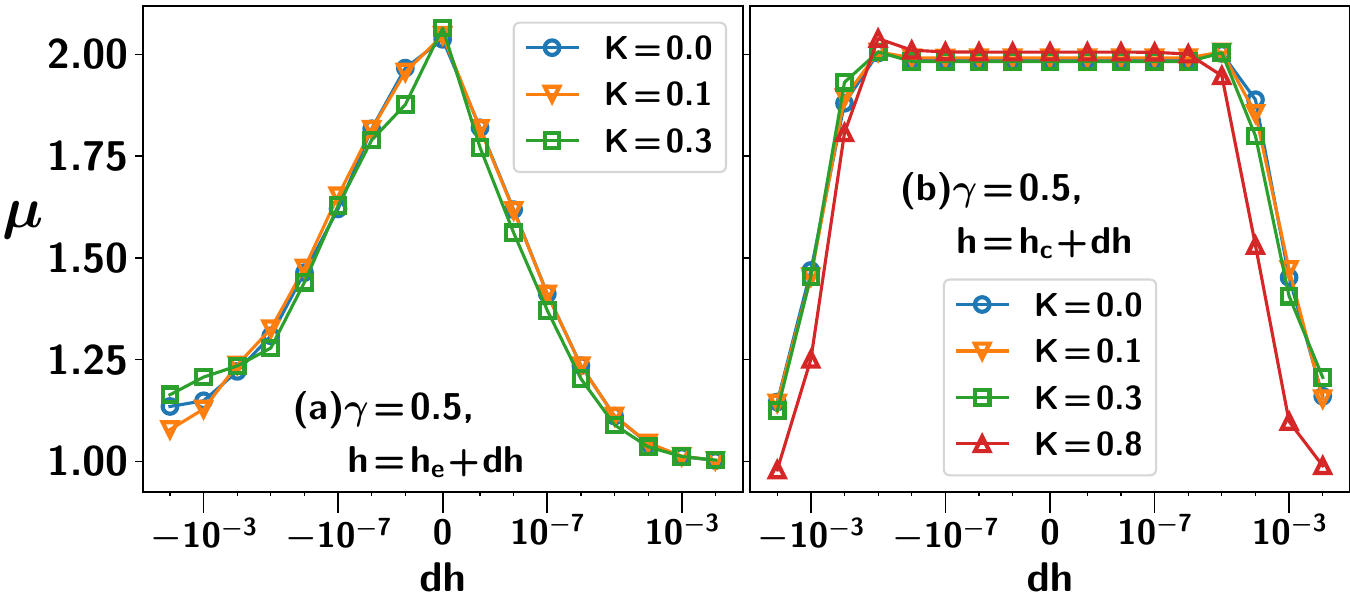}
    \caption{\textbf{The scaling exponent \(\mu\) vs small increment in the magnetic field, \(dh\).} In (a), for a fixed \(\gamma\) and different values of \(K\), we vary \(h= h_c + dh\) while in (b), \(h = h_e + dh\). In both the situations, we observe that when \(h \rightarrow h_c\) or \(h_e\), QFI scales as \(N^2\), attaining Heisenberg scaling.  All the axes are dimensionless.} 
    \label{fig:qfi_scaling}
 \end{figure}


{\it Real energy eigenvalues.} When $\epsilon^2(\phi_p)> 0$, i.e., the eigenvalues are real for momentum index $p$, 
$\mathcal{A}_p^-=2K(\gamma+K)\sin^2\phi_p+2g_p(g_p-\epsilon(\phi_p))$. Since both $u_p^-, v_p^-\in \mathbb{R}$,  $\langle\psi^-|d_h\psi^-\rangle_p=0$, and 
\begin{eqnarray}
    &&\mathcal{F}_h^{nH}(\ket{\psi^-}_p)= (u_p^-v_p^-/(\epsilon(\phi_p)\mathcal{A}^{-}_p))^2\nonumber\\
    &&=\frac{\sin^2\phi_p(\gamma^2-K^2)^2}{\epsilon^2(\phi_p)(\gamma g_p+\epsilon(\phi_p)K)^2}\equiv \mathcal{F}_h^{nH(r)}(\phi_p). 
    \label{eq:QFIreal}
\end{eqnarray}
{\it Imaginary eigenvalues.} For $\epsilon^2(\phi_p)<0$, 
the QFI for each momentum block turns out to be
\begin{equation}
    \mathcal{F}^{nH}_h(\ket{\psi^-}_p)
    =\frac{(\gamma^2-K^2)}{-\epsilon^2(\phi_p)\gamma^2}\equiv \mathcal{F}_h^{nH(im)}(\phi_p).
    \label{eq:QFIim}
\end{equation}


\textbf{Case I} : $K>\gamma$. The system remains in the unbroken phase for all values of \( h \), and the QFI is  \(\mathcal{F}^{nH}_h(\ket{\Psi}) = \sum_p \mathcal{F}_h^{nH(r)}(\phi_p)\). To achieve a higher scaling of \( \mathcal{F}_h^{nH(r)}(\phi_p) \), we note that the denominator in Eq.~(\ref{eq:QFIreal}) can approach zero for certain values of \( p \), which results in \(\mathcal{F}_h^{nH}(\ket{\Psi})\) becoming large but finite for finite system size $N$.  In this scenario as $N\to\infty$, both the conditions \( \epsilon(\phi_p) \to 0 \) and  \( \gamma g_p + \epsilon(\phi_p)K \to 0,\) are satisfied only at \( h\equiv h_c = 1 \) and \( p = N/2 - 1 \), where \( \phi_{N/2-1} \to \pi \) as \( N \) increases. This implies that   \( \Delta\phi \equiv \pi - \phi_{N/2-1} = \frac{\pi}{N} \to 0\) as \(N\to \infty\). To determine the highest scaling, we perform a Taylor series expansion around \( x_0 = \pi \) with $h_c=1$, yielding   \( \epsilon^2(x) = (K^2 - \gamma^2)(\pi - x)^2 + \mathcal{O}((\pi - x)^4)\) and \((\gamma h_c+\gamma\cos x+\epsilon(x)K)^2 = K^2 (K^2 - \gamma^2) (\pi - x)^2 + \mathcal{O}((\pi - x)^3)\), giving \(\frac{\sin^2\phi_p (K^2-\gamma^2)^2}{(\gamma h_c+\gamma\cos\phi_p+\epsilon(\phi_p)K)^2} = \frac{K^2-\gamma^2}{K^2} + \mathcal{O}\left(\frac{\pi}{N}\right)\). 
Substituting this into \( \mathcal{F}_h^{nH(r)}(\phi_p) \), we obtain  
\begin{eqnarray}
   \mathcal{F}_h^{nH(r)}(\phi_{N/2-1}) &=& \frac{\frac{K^2-\gamma^2}{K^2} + \mathcal{O}\left(\frac{\pi}{N}\right)}{(K^2-\gamma^2)\left(\frac{\pi}{N}\right)^2+\mathcal{O}\left(\frac{\pi}{N}\right)^3} \nonumber \\
    &=& \frac{1}{K^2} \left(\frac{N}{\pi}\right)^2 + \mathcal{O}(N),
\end{eqnarray}  
which clearly demonstrates that as \( h_c = 1 \), 
\(
\mathcal{F}^{nH}_{h\to1}(\ket{\Psi}) \sim N^2
\) (see also Fig. \ref{fig:qfi_c}(b)). 


\textbf{Case II} : $K<\gamma$. In this domain, the system exhibits both unbroken and broken phases. When \( h > h_e = \sqrt{\tilde{h}} \) with \(\tilde{h}=1+\gamma^2-K^2\), \(\epsilon^2(\phi_p) > 0\) for all momentum indices \( p \) , indicating that the system is in the unbroken regime. At \( h =  h_e \), the condition \(\epsilon(\phi_p) \to 0\), i.e., both real and imaginary part of \(\epsilon(\phi_p)\) approach zero. Such condition is satisfied at \( \phi_p \to \omega_c \), where \( \cos\omega_c = -h_e^{-1} \), while for finite system size \( N \), \(\epsilon^2(\phi_p) > 0\) still holds. Given that \( h_e + \tilde{h} \cos\omega_c = 0 \), a series expansion of \(\epsilon^2(x)\) around \( x_0 = \omega_c \) in Eq.~(\ref{eq:expansion}) yields   \( \epsilon^2(x) = (\gamma^2 - K^2)(x - \omega_c)^2 + \mathcal{O}((x - \omega_c)^3), \) where \( \phi_p - \omega_c \sim \frac{\pi}{N} \) for some \( p = p_c \). Since all other terms in Eq.~(\ref{eq:QFIreal}) remain finite and nonvanishing, we obtain  
\begin{equation}
    \mathcal{F}_h^{nH(r)}(\phi_{p_c}) \sim \frac{1}{\gamma^2 - K^2} \left(\frac{N}{\pi}\right)^2,
\end{equation}  
which again confirms the scaling of  \( N^2 \) at the exceptional point.  

In the broken regime, where \( h < h_e \), the function \(\epsilon(x)\) vanishes at two \( \omega_{\pm} \), satisfying the relation   \( h_e \cos\omega_{\pm} = -\left(h/h_e\right) \pm \sqrt{(\gamma^2 - K^2)(1 - (h/h_e)^2)} \). Additionally, since \(\sin\omega_{\pm} (h + \tilde{h} \cos\omega_{\pm}) \neq 0\), a series expansion of $\epsilon^2(x)$ around $\omega_{\pm}$ gives \(\epsilon^2(x) \sim (x - \omega_{\pm})\), leading to   \( \mathcal{F}_h^{nH(r,im)}(\phi_p) \sim N \quad \text{as} \quad \phi_p \to \omega_{\pm} \quad \text{with increasing } N \).  At \( h=h_c = 1 \),  \(\epsilon(\phi_{N/2-1}) \to 0\) and \(\sin\omega_-\to 0\) as \( \phi_{N/2-1} \to \pi \), leading to  
\begin{eqnarray}
    -\epsilon(\phi_{N/2-1}) \sim \left(\frac{\pi}{N}\right)^2 \implies \mathcal{F}^{nH(im)}_h(\phi_{N/2-1}) \sim N^2.
\end{eqnarray}  
Thus, a quantum advantage is achieved whenever both conditions, \(\epsilon^2(x_0) = 0\) and   \( \left. \frac{d}{dx} \epsilon^2(x) \right|_{x = x_0} = 0 \),   are satisfied for some \( x \in [0, \pi] \) which happens only at the critical  and the exceptional points (see Fig. \ref{fig:qfi_c}(a) ).

\end{proof}

In contrast to \(K>\gamma\), in the region \(K<\gamma\), we obtain criticality-enhanced scaling in both gapped to gapless transition point (see Fig. \ref{fig:qfi_scaling}(b)) as well as in EP (see Fig. \ref{fig:qfi_scaling}(a)) which highlights the usefulness of the non-Hermitian system over the Hermitian ones. At this point, we can ask -- ``{\it can this Heisenberg limit be crossed via tweaking the parameters of the non-Hermitian Hamiltonian?}"  We answer it affirmatively by proving below that we can surpass HL scaling when the system-size, \(N\), is moderate and the system parameters are tuned appropriately. 


 \begin{corollary} When \(K \rightarrow \gamma\),  super-Heisenberg scaling in QFI can be obtained provided  the system-size is moderate.
 \end{corollary}
\begin{proof}
     Let us consider the QFI in the regime $K>\gamma$, with \(K\to\gamma\), as given in Eq. (\ref{eq:QFIreal}). We use this relation and show that when \(K-\gamma=\kappa\to 0^+\), the state in which QFI is calculated has to be prepared at \(h=h_c=1\).  Note that we keep finite system size $N$ here, such that $\frac{\pi}{N}\gg\kappa$. Therefore, Taylor expansion of terms of Eq. (\ref{eq:QFIreal}) around $K=\gamma$ gives $(\gamma^2-K^2)^2=4\gamma^2\kappa^2+\mathcal{O}(\kappa^3)$, $\epsilon^{-2}(x,\kappa)=g^{-2}(x)+\mathcal{O}(\kappa)$ and $(\gamma g(x)+\epsilon(x,\kappa)K)^{-2}= \left(2\gamma g(x)\right)^{-2} + \mathcal{O}(\kappa)$, with $g(x)=h_c+\cos x$.
Therefore, taking only the leading terms in $\kappa$, the QFI is given by
       \begin{align}
           \mathcal{F}_h^{nH(r)}(\phi_p, \kappa)
         &=\frac{4\gamma^2\sin^2\phi_p\kappa^2}{4\gamma^2 g^4(\phi_p)} +\mathcal{O}(\kappa^3),
         \label{eq:superscaling}\\
         \mathcal{F}_h^{nH(r)}(\phi_{L/2-1}, \kappa)&=\frac{\sin^2(\pi-\pi/N)\kappa^2}{(1+\cos(\pi-\pi/N))^4}+\mathcal{O}(\kappa^3).\nonumber
       \end{align}
    In the last line we use the maximum contribution on the coefficient of \(\kappa^2\) from each momentum, which is at \(p=N/2-1\).
    Now using the series expansion of sine and cosine series, the denominator scales as $(\pi/N)^8$, while the numerator scales as $(\pi/N)^2$,  keeping $\pi/N\gg\kappa$ as $\phi_p\to \pi$. Therefore, $\mathcal{F}_h^{nH(r)}(\phi_p\to\pi, \kappa\to0)|_{\pi-\phi_p\gg\kappa} \sim (N/\pi)^6$, giving super-Heisenberg scaling. Note that in the thermodynamic limit, as $\phi_p\to\pi$ for arbitrary small, but finite $\kappa$, Eq. (\ref{eq:superscaling}) is not valid. For smaller values of $\kappa$, the scaling of $N^6$  is valid for larger $N$, while the Fisher information value decreases with the decrease of $\kappa$, as shown in Fig. \ref{fig:N6_scaling}. In the thermodynamic limit $\mathcal{F}_h^{nH(r)}=0$ at $h_c=1$ and $\gamma=K$, keeping the Theorem 1 valid in the thermodynamic limit. 
\end{proof}

{\bf Note 1.} While the above analysis is true for $K>\gamma$, when $K<\gamma$,  $\kappa\to 0^-$ and  $h_c=1$,  we observe that the QFI again scales as $N^6$ for moderate system-size and provides $N^2$-scaling in the thermodynamic limit.
     

{\bf Note 2.} In the Corollary, we show that when \( K \to \gamma \), at criticality (\( h_c = 1 \)), super-Heisenberg scaling is achieved. Although the value of QFI  is  very small and vanishes when \( K = \gamma \). On the other hand in case of \( K \to \gamma \), its value increases with \( N \), ultimately reaching super-Heisenberg scaling for moderate values of \( N \), as illustrated in Fig. \ref{fig:N6_scaling}. On the other hand, as the system size \( N \) increases, the scaling of the QFI saturates to the Heisenberg scaling as shown in Theorem 1. 
This is because, in the thermodynamic limit, \(\pi/N \gg \kappa\) does not hold as \(\phi_p \rightarrow \pi\). 
Note that although this result is similar to that obtained for the Hermitian model \cite{Mondal2024Jul} (see Appendix \ref{app:sup_heisenberg}), 
the super-Heisenberg scaling of the QFI in the non-Hermitian model  arises from the competition between the Hermiticity and non-Hermiticity parameters which is not the case for the Hermitian system. The above result further demonstrates that the engineered bath can be used to tune the scaling of the QFI \cite{Agarwal2023May}.

\begin{figure}
    \centering
    \includegraphics[width=\linewidth]{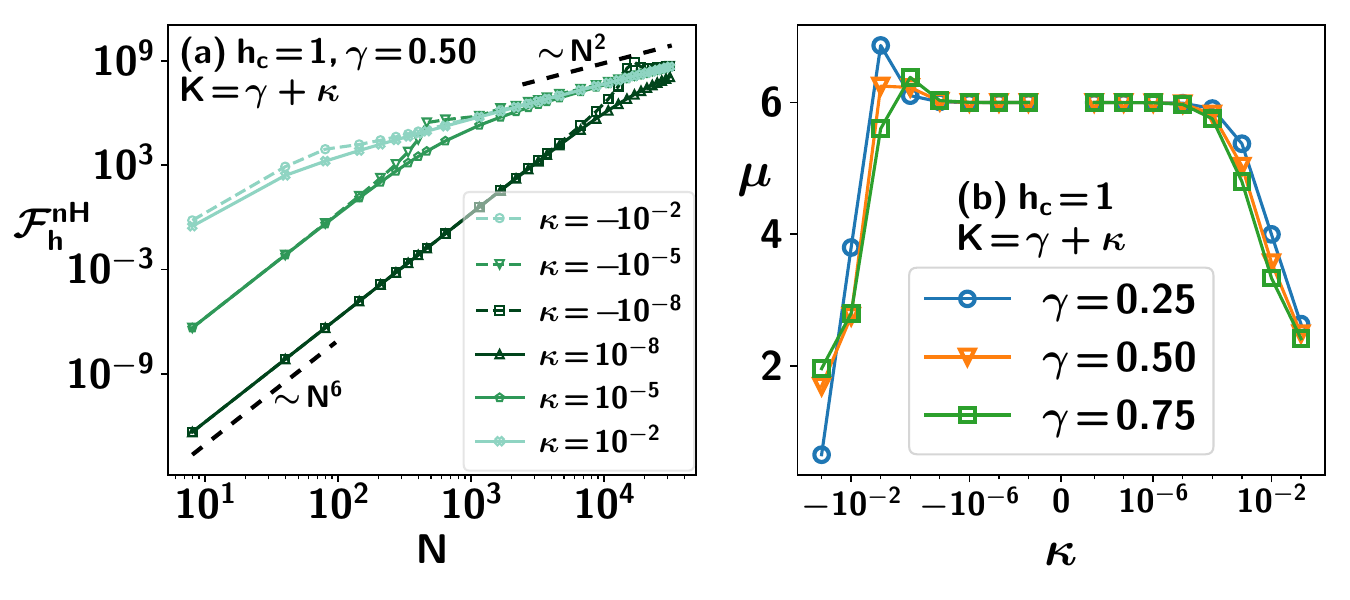}
    \caption{(a) \(\mathcal{F}_h^{nH}\) (ordinate) against N (abscissa) for different  \(K-\gamma = \kappa\) values. As seen from the figure, \(\kappa\) is small, the QFI scales as \(N^6\). (b) The scaling exponent, \(\mu\) of QFI (ordinate) with respect to \(\kappa\) (abscissa) for different \(\gamma\) values with \(h_c =1\). 
    Clearly, when \(\kappa\) small and \(N\) is moderate, the QFI beats Heisenberg scaling as proven in Corollary. Crucially, this super-Heisenberg scaling is achieved because of the system's conflict between hermiticity and non-Hermiticity features, in contrast to the Hermitian model.
    All the axes are dimensionless.} 
    \label{fig:N6_scaling}
 \end{figure}

\section{Probing  nonlinear scaling dynamically in the broken region}
\label{sec:dynamicsexceptionalpoint}

In static quantum metrology, the effectiveness of a strategy often 
depends on the system's criticality, where the quantum Fisher information exhibits Heisenberg scaling. Alternatively, a probe can be dynamically prepared, starting from an initial state that contains no information about the unknown parameters. As the system evolves under a parameter-dependent evolution operator, the QFI increases, displaying scaling behavior with respect to both the time and the system-size. 
 Let us now investigate whether non-Hermiticity can improve QFI scaling or provide atleast the same scaling as obtained via its Hermitian analogs. Note that it was previously proven \cite{ding_prl_2023} that time scaling cannot be improved with non-Hermiticity of the evolving operator which we counter in this work.

We here focus on the dynamics governed by \(\mathcal{RT}\)-symmetric non-Hermitian Hamiltonian, in Eq. (\ref{eq:Hamil}), starting
 from the initial state, \(\ket{\Psi_0}=\ket{0\ldots 0}\) 
which can be effectively prepared as the ground state of local magnetic field in the $z$-direction. 
The initial state can be written in terms of the fermionic creation and annihilation operators, i.e., $\ket{\Psi_0}=\bigoplus_p\ket{0}_p \equiv\bigoplus_p\ket{\psi_0}_p$, while 
the evolving operator in the $\{\ket{0}, a_p^\dagger a_{-p}^\dagger\ket{0}\}$ basis can be written as
\begin{equation}
    U_p(t) =\exp(-i\mathcal{H}_p t) = \cos(\epsilon_pt)\mathbb{I}_2 - i\frac{\sin(\epsilon_pt)}{\epsilon_p}\mathcal{H}_p,
\end{equation}
for each momentum index $p$, where $\mathbb{I}_2$ is an $2\times2$ identity operator. The Hermitian case can be easily retrieved by $\gamma\to-i\gamma$ which is $\mathcal{H}_p(h,-i\gamma,K)$ in Eq. (\ref{eq:ksea_P}). Unlike the Hermitian case, $U_p(t)$ is not a unitary operator as $\mathcal{H}_p$ is non-Hermitian. Hence, the dynamical state, and its differential with magnetic field $h$, is given as
\begin{gather}
    \ket{\psi_t}_p=\mathcal{A}(t)U_p(t)\ket{\psi_0}_p, \nonumber\\
    d_h\ket{\psi_t}_p = d_h\mathcal{A}(t)U_p(t) \ket{\psi_0}_p + \mathcal{A}(t) d_hU_p(t)\ket{\psi_0}_p, \nonumber
\end{gather}
where $d_h\equiv\frac{d}{dh}$ is the differential operator. The normalization of the resultant dynamical states $\ket{\psi_t}_p$ for each momentum $p$, is also performed at each time $t$, with time-dependent normalization constant, $\mathcal{A}(t)=\sqrt{_p\langle\psi_{t}\ket{\psi_t}_p}$. For the dynamical state,  we compute $\mathcal{F}^{nH}_h(\ket{\Psi_t})=\sum_p\mathcal{F}_h^{nH}(\ket{\psi_t}_p)$, with
\begin{eqnarray}
    \mathcal{F}_h^{nH}(\ket{\psi_t}_p) = 4 \mathcal{A}^2(t) \Big( {_p}\bra{\psi_0} d_h U_p^\dagger d_h U_p\ket{\psi_0}_p\nonumber \\ - \mathcal{A}^2(t)| {_p}\bra{\psi_0}  U_p^\dagger d_h U_p\ket{\psi_0}_p|^2  \Big).
\end{eqnarray}
In the dynamical scenario, both the time and the system-size can be accounted as resource, with the SQL given by $\mathcal{F}\sim Nt^2$, which can be achieved by an ensemble of $N$ non-interacting particles coherently evolving with local magnetic field from a product state \cite{puig2024arxiv}. 

\begin{figure}
    \centering
    \includegraphics[width=\linewidth]{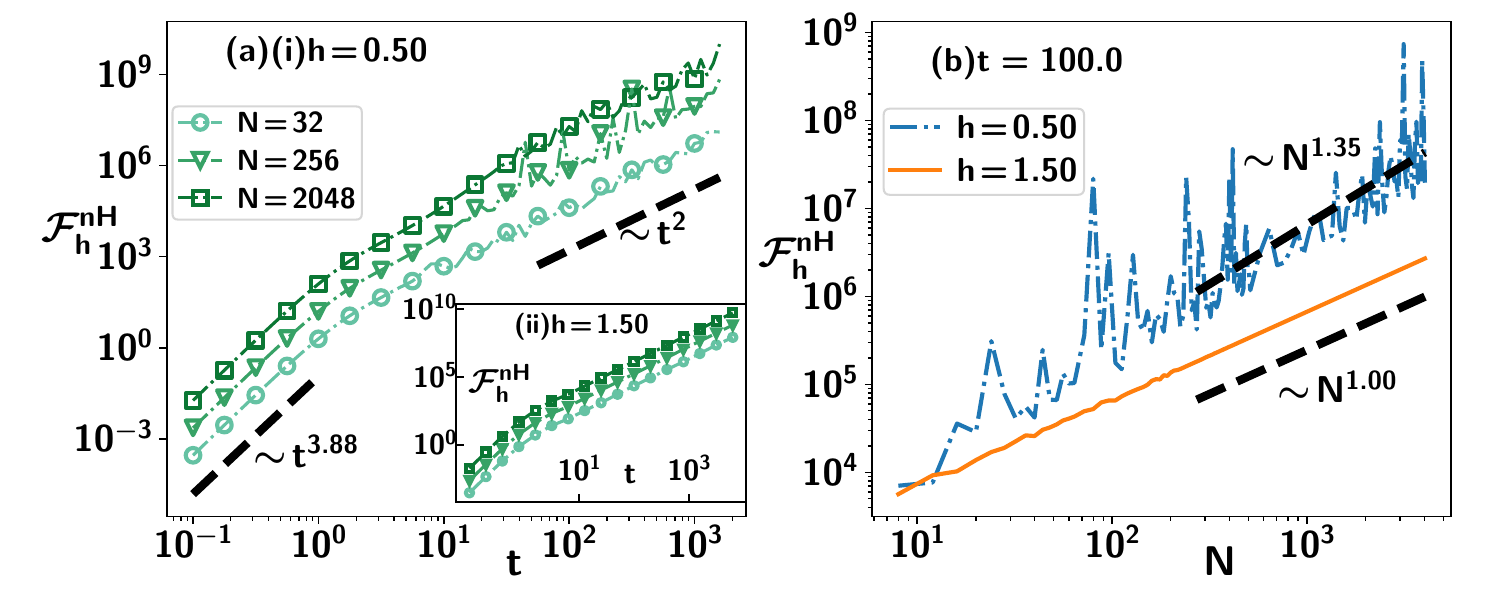}
    \caption{\textbf{Scaling of Fisher information with time and system-size.}  (a) QFI (ordinate) against time (abscissa) for different system-sizes in (i) broken  and (ii) unbroken phases. QFI goes from $t^\beta$, with $\beta\sim4$ to $\beta =2$ at large time. (b)  QFI scaling with system-size $N$ (abscissa) at $t=100$. Other parameters are \(K=0.2\) and \(\gamma=0.5\). In the unbroken phase, we obtain the SQL ($\mu=1$), whereas in the broken regime, scaling goes beyond the SQL scaling, i.e., $\mathcal{F}^{nH}_h\sim N^\mu$ with $\mu>1$. All the axes are dimensionless.} 
    \label{fig:dyn_qfi1}
 \end{figure}

In the non-Hermitian case, we observe that $\mathcal{F}^{nH}_h\sim N^\mu t^\beta$  with \(\mu >1\) and \(\beta >2\) in some situations. Specifically, 
 we find that in the unbroken phase, where the Hamiltonian is pseudo-Hermitian,  the behavior is similar to its Hermitian counterparts -- for small time, say $t<1$, $2< \beta\lesssim 4$, while at large times, SQL is retrieved  (see Fig. \ref{fig:dyn_qfi1} (a)).
 Interestingly, in the broken phase, the QFI increases non-monotonically with the system-size at large times. While at small times, $\beta\lesssim4$ and $\mu\sim 1$ is obtained, at large times,  $\mathcal{F}^{nH}_h\sim N^\mu t^2$, with $\mu>1$, which is not the case if the evolution is done with Hermitian Hamiltonian of same parameters (see Fig. \ref{fig:dyn_qfi1}(b)), highlighting the benefit of non-Hermitian system in dynamics.
 

\section{Conclusion}
\label{sec:conclu}

Quantum mechanical principles have been used to discover a number of upgraded and efficient technologies that perform better than their classical equivalents.  Notable examples include quantum sensors that promise enhanced precision in biomedical operations, quantum imaging, and the detection of magnetic fields and gravitational waves.  In order to apply quantum sensing techniques, it is necessary to build efficient protocols as well as appropriate physical platforms. In this context, non-Hermitian sensing using exceptional points has emerged as a possible method for developing quantum sensors. 

In conclusion, we identified a \(\mathcal{RT}\)-symmetric non-Hermitian Hamiltonian capable of accurately estimating an external magnetic field.  The model based on non-Hermitian \(XY\) and  Kaplan-Shekhtman-Entin-Wohlman-Aharony (KSEA) interactions among \(N\) spin-\(1/2\) particles can be experimentally realized via reservoir engineering, making it a realistic choice for investigation.  Furthermore, it can be solved analytically, allowing for investigation in both the thermodynamic limit and finite system sizes. The Hamiltonian exhibits entirely real eigenvalues within specific parameter regimes, which can be modified using a complex gauge transformation, creating a tunable transition between broken and unbroken regions controlled by the external magnetic field.

We demonstrated that the quantum Fisher information  of the ground state scales as \(\sim N^2\) at exceptional points and gapless-to-gapped region crossings in the thermodynamic limit, indicating the emergence of Heisenberg scaling. Our analytical findings are further validated by numerical simulations. Additionally, we revealed that the interplay between Hermitian and non-Hermitian parameters can lead to a super-Heisenberg scaling (\(\sim N^6\)) for finite system sizes. 

In the dynamical situation, when an initial product state is evolved according to a $\mathcal{RT}$-symmetric non-Hermitian Hamiltonian in the broken phase, we found that the sensing of local magnetic field can be enhanced with system-size, i.e. $N^\mu, \mu>1$ in long times, beating the known standard quantum limit. These results suggest that such non-Hermitian systems hold significant potential for achieving quantum advantages in sensing applications which cannot be achieved via their Hermitian counterparts.


\acknowledgments

We  acknowledge the use of the cluster computing facility at the Harish-Chandra Research Institute. This research was supported in part by the ``INFOSYS" scholarship for senior students. LGCL is funded by the European Union. Views and opinions expressed are however those of the author(s) only and do not necessarily reflect those of the European Union or the European Commission. Neither the European Union nor the granting authority can be held responsible for them.
This project has received funding from the European Union’s Horizon Europe research and innovation programme under grant agreement No 101080086 NeQST.

\appendix
\section{Exact diagonalization of \(iKSEA\) model}
\label{appendix:exactdiag}

To diagonalize the Hamiltonian,  \( H^{iKSEA}\) in Eq. \ref{eq:Hamil}, we first apply the Jordan-Wigner transformation (JW) which reduces the model  to a free-fermionic Hamiltonian, given by
\begin{eqnarray}
    H^{iKSEA}_{JW} &=& \frac{1}{2}\sum_j \left( c_j^\dagger c_{j+1} + c_{j+1}^\dagger c_{j} \right) + i(\gamma-K) c_j^\dagger c_{j+1}^\dagger \nonumber \\
    && + i(\gamma+K) c_{j+1} c_j + h(2c_j^\dagger c_j - 1),
    \label{eq:iksea_jw}
\end{eqnarray}
where \( c_j^\dagger \) and \( c_j \) are fermionic creation and annihilation operators, obeying \( \{ c_i, c_j^\dagger \} = \delta_{ij} \). This model is a variant of the Kitaev model with imbalanced pairing, exhibiting exotic topological properties \cite{kitaev_iop_2001, li_prb_2018}. 

Next, we perform the second step of the procedure, which involves a Fourier transformation,  transforming the fermionic operators into their conjugate momenta, given by
\begin{eqnarray}
    c_j &=& \frac{1}{\sqrt{N}} \sum_{p=-N/2}^{N/2} \exp \left( -\frac{2\pi jp}{N} \right) c_p, \\
    c_j^\dagger &=& \frac{1}{\sqrt{N}} \sum_{p=-N/2}^{N/2} \exp \left( \frac{2\pi jp}{N} \right) c_p^\dagger.
\end{eqnarray}
Due to the periodic boundary condition, the system exhibits translational invariance, implying that the momentum is a good quantum number. This allows us to decompose the Hamiltonian into individual momentum sectors, such that \( H^{iKSEA}_{JW} = \oplus_p H^{iKSEA}_p \). Consequently, the Hamiltonian in Eq. (\ref{eq:iksea_jw}) simplifies to
\begin{eqnarray}
    H^{iKSEA}_p &=& \sum_{p>0} \left( h + \cos\phi_p \right) \left( c_p^\dagger c_p + c_{-p}^\dagger c_{-p} \right) \nonumber \\
    && + \sin\phi_p \left[ (\gamma - K) c_p^\dagger c_{-p}^\dagger + (\gamma + K) c_p c_{-p} \right] - h, \nonumber
    \label{pksea_ixy}
\end{eqnarray}
where \( \phi_p = \frac{(2p-1)\pi}{N} \) and \( p \in \{1, \dots, N/2\} \) as we adopt the anti-periodic boundary condition \( c_{N+1} = -c_1 \) \cite{santoro_ising_beginners_2020}. In the thermodynamic limit \( N \to \infty \), the momentum becomes continuous, and \( \phi_p \in (0, \pi) \).

This Hamiltonian can be expressed as \( H_p^{iKSEA} = \hat{\Gamma}_p^\dagger \mathcal{H}_p \hat{\Gamma}_p \), where \( \hat{\Gamma}_p = (c_{-p}^\dagger, c_p)^T \) is the Nambu spinor \cite{santoro_ising_beginners_2020}, and
\begin{equation}
    \mathcal{H}_p = \left[ \begin{array}{cc}
    -h - \cos \phi_p & -(\gamma + K) \sin \phi_p \\
    (\gamma - K) \sin \phi_p & \cos \phi_p + h
    \end{array} \right].
\label{eq:ksea_P_app}
\end{equation}
The corresponding eigenvalues of this model are \( E_p^{iKSEA} = \pm \epsilon (\phi_p) \), where
\begin{equation}
    \epsilon (\phi_p) = \sqrt{(h + \cos\phi_p)^2 + ( K^2-\gamma^2 )\sin^2\phi_p},
    \label{eq:dispersion_relation_app}
\end{equation}
which can be used to obtain the phase diagram of the model \cite{Agarwal2023May}. In our case, we have used the above equation to prove Theorem 1 and the Corollary 1.  

\section{Resemblance of super-Heisenberg scaling with Hermitian spin chain}
\label{app:sup_heisenberg}
The emergence of super-Heisenberg scaling in the limit \(K \to \gamma\) mirrors the behavior observed in Hermitian Hamiltonians, which can be understood from an alternative perspective. Consider a gauge transformation applied to the fermionic operators of the form:
\begin{equation}
    c_j = e^{\mu/2} e^{i\theta/2} c_j, \quad \text{and} \quad \bar{c}_j = e^{-\mu/2} e^{-i\theta/2} c_j^\dagger,
    \label{eq:map_xy_ksea}
\end{equation}
where \(e^{\mu} = \sqrt{\frac{K - \gamma}{K + \gamma}}\) and \(\theta = -\pi/2\), with the anticommutation relation \(\{c_j, \bar{c}_k\} = \delta_{jk}\). Under this transformation, Eq.~(\ref{eq:iksea_jw}) can be reformulated as
\begin{eqnarray}
    \nonumber H_{XY}^{JW} &=& \frac{1}{2} \sum_j \left( \bar{c}_j c_{j+1} + \bar{c}_{j+1} c_j \right) + \gamma' \left( \bar{c}_j \bar{c}_{j+1} + c_{j+1} c_j \right) \\
    && + h (2 \bar{c}_j c_j - 1),
    \label{eq:xy_jw}
\end{eqnarray}\\
where \(\gamma' = \sqrt{K^2 - \gamma^2}\). This Hamiltonian corresponds to the free-fermionic representation of the Hermitian \(XY\) spin model in a transverse magnetic field, and is also analogous to a one-dimensional Kitaev chain with a modified superconducting pairing term. Notably, as \(K \to \gamma\), we have \(\gamma' \to 0\), indicating a slight breaking of the \(\mathbb{U}(1)\) symmetry, resembling the behavior in the Hermitian case~\cite{Mondal2024Jul}. A similar pattern emerges in the non-Hermitian setting due to the interplay and competition between Hermitian and non-Hermitian contributions.

\section{Derivation of Quantum Fisher Information for ground state of iKSEA model}
\label{app:qfi_proof}
The normalized ground state of the \(2 \times 2\) Hamiltonian in momentum space, as defined in Eq.~(\ref{eq:ksea_P_app}), is given by
\[
\ket{\psi^-}_p = \frac{1}{\sqrt{\mathcal{A}_p^-}} \begin{bmatrix} u_p^- \\ v_p^- \end{bmatrix},
\]
where \( u_p^- = \alpha_p^+ \) and \( v_p^- = \epsilon(\phi_p) - g_p \). Here, the parameters are defined as \( \alpha_p^\pm = (\gamma \pm K) \sin\phi_p \), and \( g_p = h + \cos\phi_p \), with \( \mathcal{A}_p^- \) being the Dirac normalization constant. The state \( \ket{\psi^-}_p \) has eigenvalue \( -\epsilon(\phi_p) \), where \( \epsilon(\phi_p) \) is given in Eq.~(\ref{eq:dispersion_relation_app}). 

The derivative with respect to \( h \), denoted by \( d_h \equiv \frac{d}{dh} \), satisfies
\( d_h \epsilon(\phi_p) = \frac{g_p}{\epsilon(\phi_p)} \), \( d_h u_p^- = 0 \) (since \( d_h \alpha_p^\pm = 0 \)), and \( d_h v_p^- = -\frac{v_p^-}{\epsilon(\phi_p)} \), as \( d_h g_p = 1 \).

Since \( \epsilon^2(\phi_p) \in \mathbb{R} \), the energy for each momentum mode is either real or purely imaginary.

\textbf{Case I:} \( \epsilon^2(\phi_p) > 0 \). In this case, \( \epsilon^*(\phi_p) = \epsilon(\phi_p) \), and both \( u_p^- \) and \( v_p^- \) are real. Thus, the normalization factor becomes
\[
\mathcal{A}_p^- = (\alpha_p^+)^2 + (\epsilon(\phi_p) - g_p)^2,
\]
and its derivative is given by
\[
d_h \mathcal{A}_p^- = -2 \frac{(v_p^-)^2}{\epsilon(\phi_p)}.
\]

The derivative of the ground state with respect to \( h \) is then
\[
\ket{d_h \psi^-}_p = \frac{1}{\sqrt{\mathcal{A}_p^-}} \begin{bmatrix} 0 \\ -\frac{v_p^-}{\epsilon(\phi_p)} \end{bmatrix} + \frac{(v_p^-)^2}{\epsilon(\phi_p) \mathcal{A}_p^-} \ket{\psi^-}_p.
\]

Since \( {_p}\langle \psi^- | d_h \psi^- \rangle_p = 0 \), the non-Hermitian quantum Fisher information reads as
\begin{align*}
\mathcal{F}_h^{nH}(\ket{\psi^-}_p) &= 4\, {_p}\langle d_h \psi^- | d_h \psi^- \rangle_p \\
&= \left( \frac{u_p^- v_p^-}{\epsilon(\phi_p) \mathcal{A}_p^-} \right)^2 \\
&= \frac{1}{\epsilon^2(\phi_p)} \left( \frac{u_p^-}{v_p^-} + \frac{v_p^-}{u_p^-} \right)^{-2}\\&=\frac{\sin^2\phi_p(\gamma^2-K^2)^2}{\epsilon^2(\phi_p)(\gamma g_p+\epsilon(\phi_p)K)^2},
\end{align*}
where we use \(\frac{u_p^-}{v_p^-} = \frac{\alpha_p^+}{\epsilon(\phi_p) - g_p} = -\frac{\epsilon(\phi_p) + g_p}{\alpha_p^-}\).


\textbf{Case II} : $\epsilon^2(\phi_p)<0$. Therefore, we have  $\epsilon^*(\phi_p)=-\epsilon(\phi_p)$ and $u_p^--\in \mathbb{R}$, giving $\mathcal{A}_p^-=(\alpha^+_p)^2 + g_p^2-\epsilon^2(\phi_p)=2\gamma\alpha^+_p\sin\phi_p$ and $d_h\mathcal{A}_p^-= 0$, which makes the calculations easier. Therefore, $\ket{d_h\psi^-}_p = \frac{-1}{\epsilon(\phi_p)\sqrt{\mathcal{A}_p^-}}[0, v_p]^T$, giving ${_p}\langle\psi^-|d_h\psi^-\rangle_p=\frac{K-\gamma}{2\gamma\epsilon(\phi_p)}$ and ${_p}\langle d_h \psi^- | d_h \psi^- \rangle_p = \frac{K-\gamma}{2\gamma\epsilon^2(\phi_p)}$. Hence, we obtain
\begin{eqnarray*}
    \mathcal{F}^{nH}_h(\ket{\psi^-}_p) &=& 4\left(\frac{K-\gamma}{2\gamma\epsilon^2(\phi_p)} - \frac{(K-\gamma)^2}{4\gamma^2(-\epsilon^2(\phi_p))}\right) \\
    &=&\frac{(\gamma^2-K^2)}{-\epsilon^2(\phi_p)\gamma^2}.
\end{eqnarray*}

\bibliography{ref.bib}

\end{document}